\documentclass[12pt]{article}
\usepackage{lMac}
\usepackage{becMac}

\def\showintremarks{n}   
\def\showissues{n}       

\begin{document}
\title{Power Series Representations for Complex Bosonic Effective Actions.\\
        \Large III. Substitution and Fixed Point Equations}

\author{Tadeusz Balaban}
\affil{\small Department of Mathematics \authorcr
       Rutgers, The State University of New Jersey \authorcr
       tbalaban@math.rutgers.edu\authorcr
       \  }

\author{Joel Feldman\thanks{Research supported in part by the Natural 
                Sciences and Engineering Research Council 
                of Canada and the Forschungsinstitut f\"ur 
                Mathematik, ETH Z\"urich.}}
\affil{Department of Mathematics \authorcr
       University of British Columbia \authorcr
       feldman@math.ubc.ca \authorcr
       http:/\hskip-3pt/www.math.ubc.ca/\squig feldman/\authorcr
       \  }

\author{Horst Kn\"orrer}
\author{Eugene Trubowitz}
\affil{Mathematik \authorcr
       ETH-Z\"urich \authorcr
       knoerrer@math.ethz.ch, trub@math.ethz.ch \authorcr
       http:/\hskip-3pt/www.math.ethz.ch/\squig knoerrer/}


\maketitle

\begin{abstract}
\noindent
In \cite{CPR,CPC,CPS} we developed a polymer--like expansion
that applies when the (effective) action in a functional integral
is an analytic function of the fields being integrated. Here,
we develop methods to aid the application of this technique
when the method of steepest descent is used to analyze the functional
integral. We develop a version of the Banach fixed point theorem that
can be used to construct and control the critical fields, as
analytic functions of external fields, and substitution formulae
to control the change in norms that occurs when one replaces
the integration fields by the sum of the critical fields and
the fluctuation fields.  

\end{abstract}

\newpage
\tableofcontents

\newpage
\section{Introduction}
In \cite{CPR,CPC,CPS}, we developed a power series 
representation, norms and estimates for an effective
action of the form
\begin{equation*}
\ln\frac{ \int e^{f(\al_1,\cdots,\al_s;z^*,z)}\,d\mu(z^*,z)}
   {\int e^{f(0,\cdots,0;z^*,z)}\,d\mu(z^*,z)}
\end{equation*}
Here, $\,f(\al_1,\cdots,\al_s;z_*,z)\,$ is an analytic function of the 
complex fields $\al_1(\bx)$, $\cdots$, $\al_s(\bx)$, $z_*(\bx)$, $z(\bx)$ 
indexed by $\bx$ in a finite set $X$, and $d\mu(z^*,z)$ is a compactly 
supported product measure. This framework has been used in \cite{UV}.

In \cite{PAR1,PAR2} we combine these power series methods with the technique of 
the block spin renormalization group for functional integrals 
\cite{KAD,BalLausane,GK,BalPalaiseau,Dim1} to see, for a many particle system of 
weakly interacting Bosons in three space dimensions, the 
formation of a potential well of the type that typically leads to 
symmetry breaking in the thermodynamic limit. 
(For an overview, see \cite{ParOv}.)
A basic ingredient of  
this block spin/functional integral approach is a stationary phase 
argument for the effective actions. For this, it is necessary to 
construct and analyze ``critical fields'' at each step. These 
critical fields are themselves functions of some external fields. 
The ``background fields'' of the block spin approach arise as 
compositions of critical fields at several renormalization group steps 
and are also functions of some external fields.

In our construction \cite{PAR1,PAR2}, the ``background fields'' and 
``critical fields'' are analytic maps that are defined on a 
neighbourhood of the origin 
in an appropriate Hilbert space of fields and that take values in 
another Hilbert space of fields. We call such objects ``field maps''. 
See Definition \ref{defSUBkrnel},  where we also generalize the definition 
of the norm of a (complex valued) analytic function of fields 
\cite[Definition \defCPCnorms]{CPC} to field maps.

In \S\:\ref{secSUBsub} we prove bounds on compositions like 
\begin{equation*}
\tilde h(\al_1,\cdots,\al_s) 
= h\Big( A_1(\al_1,\cdots,\al_s),\cdots, A_r(\al_1,\cdots,\al_s)\Big)
\end{equation*}
in terms of bounds on $h$ and the $A_j$'s. Here, $h$ is a 
function of $r$ fields and $A_1,\cdots, A_r$ are field maps. See 
Proposition \ref{propSUBsubstitution} and Corollary \ref{corSUBsubstitution}.

The critical fields for each block spin renormalization group
transformation are critical configurations for some action. 
The equations that determine these critical configurations
can be expressed as (systems of) implicit equations of the type
\begin{equation*}
\ga = F(\al_1,\cdots,\al_s;\ga)
\end{equation*}
which have to be solved for $\ga$ as a function $\al_1,\cdots,\al_s$.
In \S\ref{secSUBsolve}, we prove the existence and uniqueness of, 
and bounds on, solutions to systems of equations of that type. 
See Proposition  \ref{propSUBeqnsoln}.

\newpage
\section{Field Maps}\label{secSUBfieldmaps}
For an abstract framework, we consider analytic functions 
$f(\al_1,\cdots,\al_s)$ of the complex fields $\al_1,\cdots,\al_s$ 
(none of which are ``history'' or source fields, in the terminology of
\cite{CPC}) on a finite set $X$. Here are some associated definitions
and notation from \cite{CPC}.

\begin{definition}[$n$--tuples]\label{defSUBntuples} 
\ 
\begin{enumerate}[label=(\alph*), leftmargin=*]
\item 
Let $n\in\bbbz$ with $n\ge 0$ and $\vec\bx=(\bx_1,\cdots,\bx_n)\in
X^n$ be an ordered $n$--tuple of points of $X$. 
We denote by $n(\vec\bx)=n$ the number of components of $\vec\bx$.
Set 
$
\al(\vec\bx)=\al(\bx_1)\cdots \al(\bx_n).
$
If $n(\vec\bx)=0$, then $\al(\vec\bx)=1$.

\item 
For each $s\in\bbbn$, we denote\footnote{
    We distinguish between $X^{n_1}\times\cdots\times X^{n_s}$
    and $X^{n_1+\cdots+n_s}$. We use $X^{n_1}\times\cdots\times X^{n_s}$
    as the set of possible arguments for $\al_1(\vec\bx_1)\cdots\al_s(\vec\bx_s)$,
    while $X^{n_1+\cdots+n_s}$ is the set of possible arguments for
    $\al_1(\vec\bx_1\circ\cdots\circ\vec\bx_s)$, where $\circ$ is the concatenation
    operator of part (c).}
\begin{equation*}
\bX^{(s)}=\bigcup_{n_1,\cdots,n_s\ge 0}X^{n_1}\times\cdots\times X^{n_s}
\end{equation*}
If $(\vec\bx_1,\cdots,\vec\bx_{s-1})\in \bX^{(s-1)}$ then
$(\vec\bx_1,\cdots,\vec\bx_{s-1},-)$ denotes the element of $\bX^{(s)}$
having $n(\vec\bx_s)=0$. In particular, $X^0=\{-\}$ and $\al(-)=1$.

\item 
We define the concatenation of $\vec\bx=(\bx_1,\cdots,\bx_n)\in X^n$ 
and $\vec\by=(\by_1,\cdots,\by_m)\in X^m$ to be 
\begin{equation*}
\vec \bx\circ\vec \by=\big(\bx_1,\cdots,\bx_n,\by_1,\cdots,\by_m)\in X^{n+m}
\end{equation*}
For $(\vec\bx_1,\cdots,\vec\bx_s),\ 
(\vec\by_1,\cdots,\vec\by_s)\in \bX^{(s)}$
\begin{equation*}
(\vec\bx_1,\cdots,\vec\bx_s)\circ (\vec\by_1,\cdots,\vec\by_s)
=(\vec\bx_1\circ\vec\by_1,\cdots,\vec\bx_s\circ\vec\by_s)
\end{equation*}
\end{enumerate}
\end{definition}

\begin{definition}[Coefficient Systems]\label{defSUBcoeffssys}
\ 
\begin{enumerate}[label=(\alph*), leftmargin=*]
\item 
A coefficient system of length $s$  is a function
$a(\vec\bx_1,\cdots,\vec\bx_s)$ which assigns a complex number to each
$(\vec\bx_1,\cdots,\vec\bx_s)\in\bX^{(s)}$. It is called symmetric if,
for each $1\le j\le s$, $a(\vec\bx_1,\cdots,\vec\bx_s)$ is invariant under
permutations of the components of $\vec\bx_j$.

\item
Let  $f(\al_1,\cdots,\al_s)$ be a function which is defined and analytic on
a neighbourhood of the origin in $\bbbc^{s|X|}$.  Then $f$ has a unique
expansion of the form
\begin{equation*}
f(\al_1,\cdots,\al_s)
=\sum_{(\vec\bx_1,\cdots,\vec\bx_s)\in\bX^{(s)}}
a(\vec\bx_1,\cdots,\vec\bx_s)\ \al_1(\vec\bx_1)\cdots \al_s(\vec\bx_s)
\end{equation*}
with $a(\vec\bx_1,\cdots,\vec\bx_s)$ a symmetric coefficient system.
This coefficient system is called the symmetric coefficient system of $f$.

\end{enumerate}
\end{definition}

We assume that we are given a metric $d$ on a finite set $X$ and 
\emph{constant} weight factors $\ka_1,\cdots,\ka_s$. 
In this environment \cite[Definition \defCPCnorms]{CPC}, 
for the norm of the function 
\begin{equation*}
f(\al_1,\cdots,\al_s)
=\sum_{(\vec\bx_1,\cdots,\vec\bx_s)\in\bX^{(s)}}
a(\vec\bx_1,\cdots,\vec\bx_s)\ \al_1(\vec\bx_1)\cdots \al_s(\vec\bx_s)
\end{equation*}
with $a(\vec\bx_1,\cdots,\vec\bx_s)$ a symmetric coefficient system, 
simplifies to
\begin{equation}\label{eqnSUBsimplenorm}
\|f\|_w=\big|a(-)\big|
+\sum_{\atop{n_1,\cdots,n_s\ge 0}{n_1+\cdots+n_s\ge 1}}\max_{\bx\in X}
\max_{\atop{1\le j\le s}{n_j\ne 0}}\max\limits_{1\le i\le n_j}
\sum_{\atop{  \atop{\vec\bx_\ell\in X^{n_\ell}}{1\le\ell\le s} }
           { {(\vec\bx_j)}_i=\bx}  }
\big|a(\vec\bx_1,\cdots,\vec\bx_s)\big|\ka_1^{n_1}\cdots\ka_s^{n_s}
e^{\tau_d(\vec\bx_1,\cdots,\vec\bx_s)}
\end{equation}
where $\tau_d(\vec\bx_1,\cdots,\vec\bx_s)$ denotes the length of the 
shortest tree in $X$ whose set of vertices contains all of the 
points in the $\vec x_j$'s. The family of functions 
\begin{equation*}
w(\vec\bx_1,\cdots,\vec\bx_s)
=\ka_1^{n(\vec\bx_1)}\cdots\ka_s^{(\vec\bx_s)}
e^{\tau_d(\vec\bx_1,\cdots,\vec\bx_s)}
\end{equation*}
is called the weight system with metric $d$ that associates 
the weight factor $\ka_j$ to the field $\al_j$. 

We need to extend these definitions to functions $A(\al_1,\cdots,\al_s)$
that take values in $\bbbc^X$, rather than $\bbbc$. That is, which map
fields $\al_1,\cdots,\al_s$ to another field $A(\al_1,\cdots,\al_s)$.
A trivial example would be $A(\al)(\bx)=\al(\bx)$.

\begin{definition}\label{defSUBkrnel}
\ 
\begin{enumerate}[label=(\alph*), leftmargin=*]
\item
An $s$--field map kernel is a function
\begin{equation*}
A:(\bx;\vec\bx_1,\cdots,\vec\bx_s)\in X\times\bX^{(s)}\mapsto
A(\bx;\vec\bx_1,\cdots,\vec\bx_s)\in\bbbc
\end{equation*}
which obeys $A(\bx;-,\cdots,-)=0$ for all $\bx\in X$.

\item 
If $A$ is an $s$--field map kernel, we define the
``$s$--field map'' $(\al_1,\cdots,\al_s)\mapsto A(\al_1,\cdots,\al_s)$ by
\begin{equation*}
A(\al_1,\cdots,\al_s)(\bx)=\sum_{(\vec\bx_1,\cdots,\vec\bx_s)\in\bX^{(s)}}
A(\bx;\vec\bx_1,\cdots,\vec\bx_s)\ \al_1(\vec\bx_1)\cdots \al_s(\vec\bx_s)
\end{equation*}

\item
We define the norm $\tn A\tn_w$ of the $s$--field map kernel
$A$ by
\begin{equation*}
\tn A\tn_w=
    \sum_{\atop{n_1,\cdots,n_s\ge 0}{n_1+\cdots+n_s\ge 1}}
    \big\|A\big\|_{w;n_1,\cdots,n_s}
\end{equation*}
where
\begin{align*}
\big\|A\big\|_{w;n_1,\cdots,n_s}
  =\max\big\{ L(A;w;n_1,\cdots,n_s)\,,\,R(A;w;n_1,\cdots,n_s)\big\}
\end{align*}
and
\begin{align*}
L(A;w;n_1,\cdots,n_s)&=\max_{\bx\in X}
     \sum_{\atop{\vec\bx_\ell\in X^{n_\ell}}{1\le\ell\le s}}
     \big|A(\bx;\vec\bx_1,\cdots,\vec\bx_s)\big|
     \ka_1^{n_1}\cdots\ka_s^{n_s}
      e^{\tau_d(\bx,\vec\bx_1,\cdots,\vec\bx_s)}\\
R(A;w;n_1,\cdots,n_s)&=\max_{\bx'\in X}
    \max_{\atop{1\le j\le s}{n_j\ne 0}}
    \max\limits_{1\le i\le n_j}\hskip-1pt
    \sum_{\bx\in X}\hskip-1pt
    \sum_{\atop{ \atop{\vec\bx_\ell\in X^{n_\ell}}{1\le\ell\le s} }
           { {(\vec\bx_j)}_i=\bx'} }\hskip-3pt
\big|A(\bx;\vec\bx_1,\cdots,\vec\bx_s)\big|
      \ka_1^{n_1}\cdots\ka_s^{n_s} \\\noalign{\vskip-0.4in}&\hskip3.4in
       e^{\tau_d(\bx,\vec\bx_1,\cdots,\vec\bx_s)}
\end{align*}
We also denote the norm of the corresponding $s$--field map 
$A(\al_1,\cdots,\al_s)$ by $\tn A\tn_w$.
\end{enumerate}
\end{definition}

\begin{remark}\label{remSUBsubtofn}
We associate to each $s$--field map kernel $A$ the analytic function
\begin{align*}
f_A(\be;\al_1,\cdots,\al_s)
&=\sum_{\bx\in X} \be(\bx)A(\al_1,\cdots,\al_s)(\bx)\cr
&=\sum_{\atop{(\vec\bx_1,\cdots,\vec\bx_s)\in\bX^{(s)}}{\bx\in X}}
A(\bx;\vec\bx_1,\cdots,\vec\bx_s)\ \be(\bx)\al_1(\vec\bx_1)\cdots \al_s(\vec\bx_s)
\end{align*}
Denote by $\hat w$ the weight system with metric $d$ that associates the
weight factor $\ka_j$ to $\al_j$, for each $1\le j\le s$, and 
the weight factor $1$ to $\be$. Then
\begin{equation*}
\|f_A\|_{\hat w}=\tn A\tn_w
\end{equation*}
\end{remark}

\begin{lemma}[Young's Inequality]\label{lemSUBLp}
Let  $d_1,\cdots,d_s\ge 0 $ be integers.

\begin{enumerate}[label=(\alph*), leftmargin=*]
\item
 Let $f(\al_1,\cdots,\al_s)$ be a function which is defined and analytic on
a neighbourhood of the origin in $\bbbc^{s|X|}$ and is of degree at least
$d_i$ in the field $\al_i$. Furthermore let $p_1,\cdots,p_s \in (0,\infty]$
be such that $\ \sum\limits_{j=1}^s \sfrac{d_j}{p_j} =1$.
Then, for all fields $\al_1,\cdots,\al_s$ such
that  $|\al_j(\bx)|\le\ka_j$ for all $\bx\in X$ and $1\le j\le s$,
\begin{equation*}
\big| f(\al_1,\cdots,\al_s) \big| 
\le \|f\|_w \,\smprod_{j=1}^s \big( \sfrac{1}{\ka_j}  \|\al_j\|_{p_j}\big)^{d_j}
\end{equation*}
where
$\
\|\al\|_p =\big(  \sum\limits_{x\in X} |\al(x)|^p \big)^{1/p}
\ $
denotes the $L^p$ norm of $\al$.

\item
Let
$(\al_1,\cdots,\al_s)\mapsto A(\al_1,\cdots,\al_s)$ 
be an $s$--field map which is of degree at least
$d_i$ in the field $\al_i$. Furthermore let $p, p_1,\cdots,p_s \in (0,\infty]$
be such that 
$\ \sum\limits_{j=1}^s \sfrac{d_j}{p_j} =\sfrac{1}{p}$.
Then, for fields $\al_1,\cdots,\al_s$ such
that  $|\al_j(\bx)|\le\ka_j$ for all $\bx\in X$ and $1\le j\le s$, the
$L^p$ norm of the field $A(\al_1,\cdots,\al_s)$ is bounded by
\begin{equation*}
\big\| A(\al_1,\cdots,\al_s) \big\|_p 
\le \tn A\tn_w \,\smprod_{j=1}^s \big( \sfrac{1}{\ka_j}  \|\al_j\|_{p_j}\big)^{d_j}
\end{equation*}
In particular
\begin{equation*}
\max_{\bx\in X}\big|A(\al_1,\cdots,\al_s)(\bx)\big|\le \tn A\tn_w
\end{equation*}
\end{enumerate}
\end{lemma}

\begin{proof}
(a)
By the definition \eqref{eqnSUBsimplenorm} of $ \|f\|_w$, we may assume that
$f$ is of the form
\begin{equation*}
f(\al_1,\cdots,\al_s)
=\sum_{    \atop{\vec\bx_\ell\in X^{n_\ell}}{1\le\ell\le s}    }
a(\vec\bx_1,\cdots,\vec\bx_s)\ \al_1(\vec\bx_1)\cdots \al_s(\vec\bx_s)
\end{equation*}
with  a symmetric coefficient $a$ and $n_\ell \ge d_\ell$.
Now apply Lemma \ref{lemSUBgenLoneLinfty} with
$\ K = a \prod\limits_{j=1}^s \ka_j^{d_j} \ $,
where we use the $L^{p_j}$ norm for the first $d_j$ components of the variable
$\vec x_j$, and the $L^\infty$ norm for the last $n_j-d_j$ components of
this variable.

\Item (b)
As in Remark \ref{remSUBsubtofn} set
\begin{equation*}
f_A(\be;\al_1,\cdots,\al_s)
=\sum_{\bx\in X} \be(\bx)A(\al_1,\cdots,\al_s)(\bx)
\end{equation*}
 As in \cite[Theorem 4.2]{LL} choose
\begin{equation*}
\be(\bx) = e^{-i\theta(x)} |A(\al_1,\cdots,\al_s)(\bx)|^{p/p'}
\end{equation*}
where $\theta(x)$ is defined by 
$\
A(\al_1,\cdots,\al_s)(x)= e^{i\theta(x)} |A(\al_1,\cdots,\al_s)(x)|
\ $
and 
$\ \sfrac{1}{p} + \sfrac{1}{p'} =1
$.
By part (a) and Remark \ref{remSUBsubtofn}
\begin{align*}
 \big\|A(\al_1,\cdots,\al_s)\big\|_p^p &= \big|f_A(\be;\al_1,\cdots,\al_s)\big| 
\le  \tn A\tn_w \,  \|\be\|_{p'}
\smprod_{j=1}^s \big( \sfrac{1}{\ka_j}  \|\al_j\|_{p_j}\big)^{d_j}\\
&=  \tn A\tn_w \,  \big\|A(\al_1,\cdots,\al_s)\big\|_p^{p/p'}
\smprod_{j=1}^s \big( \sfrac{1}{\ka_j}  \|\al_j\|_{p_j}\big)^{d_j}
\end{align*}
\end{proof}

\begin{remark}\label{remSUBlinear}
A linear map $L:\bbbc^X\rightarrow\bbbc^X$ can be thought of
as a $1$--field map kernel. The relation between the norm $\tn L\tn_w$
as a field map kernel and the norm $\tn L\tn$ as in
\cite[Definition \defCPCtriplenorm]{CPC} is
\begin{equation*}
\tn L\tn_w=\ka_1\tn L\tn
\end{equation*}
The field $L(\al_1)$ is
\begin{equation*}
L(\al_1)(\bx)=\sum_{\by\in X}L(\bx,\by)\al_1(\by)
\end{equation*}
\end{remark}

\begin{remark}\label{remSUBdifferentLattices}
In Definition \ref{defSUBkrnel}, we have assumed, for simplicity,
that the field map $A$ maps fields $\al_1,\cdots,\al_s$ on a set $X$ 
to a field $A(\al_1,\cdots,\al_s)$ on the same set $X$. We will apply
this definition and the results later in this paper when the input fields
$\al_1,\cdots,\al_s$ are defined on a subset $X_1\subset X$
and the output field $A(\al_1,\cdots,\al_s)$ is defined on a, possibly
different, subset $X_2\subset X$. We extend Definition \ref{defSUBkrnel} and
the results later in this paper to cover this setting by viewing
$\al_1,\cdots,\al_s$
and $A(\al_1,\cdots,\al_s)$ to be fields on $X$ --- set 
$\al_1,\cdots,\al_s$ to zero on $X\setminus X_1$ and $A(\al_1,\cdots,\al_s)$
to zero on $X\setminus X_2$.
\end{remark}

\newpage
\section{Substitution}\label{secSUBsub}

We now proceed to prove bounds on compositions like 
\begin{equation*}
\tilde h(\al_1,\cdots,\al_s) 
= h\Big( A_1(\al_1,\cdots,\al_s),\cdots, A_r(\al_1,\cdots,\al_s)\Big)
\end{equation*}
in terms of bounds on $h$ and the $A_j$'s.

\begin{lemma}\label{lemSUBdiff}
Let  $\la_1,\ \cdots,\ \la_s$ be constant weight factors and
let  $w_\de$ be the weight system with metric $d$ that associates 
the weight factor $\ka_j$ to $\al_j$ and $\la_j$ to a field $\de_j$. 
Fix any $\si\ge 1$ and
let  $w_\si$ be the weight system with metric $d$ that associates 
the weight factor $\ka_j+\si\la_j$ to $\al_j$. 
\begin{enumerate}[label=(\alph*), leftmargin=*]
\item
Let $f(\al_1,\cdots,\al_s)$ be an analytic function on a
neighbourhood of the origin in $\bbbc^{s|X|}$. Set
\begin{equation*}
\de f\big(\al_1,\cdots,\al_s,\de_1,\cdots,\de_s\big)
=  f\big(\al_1+\de_1,\cdots,\al_s+\de_s\big)
- f\big(\al_1,\cdots,\al_s\big)
\end{equation*}
Then
\begin{equation*}
\|\de f\|_{w_\de}\le \sfrac{1}{\si} \|f\|_{w_\si}
\end{equation*}
More generally, if $p\in\bbbn$ and
$\de f^{(\ge p)}\big(\al_1,\cdots,\al_s,\de_1,\cdots,\de_s\big)$ 
is the part of $\de f$
that is of degree at least $p$ in $\big(\de_1,\cdots,\de_s\big)$, then
\begin{equation*}
\|\de f^{(\ge p)}\|_{w_\de}\le \sfrac{1}{\si^p} \|f\|_{w_\si}
\end{equation*}

\item
Let $A$ be an $s$--field map and define the $2s$--field map 
$\de A$ by
\begin{equation*}
\de A\big(\al_1,\cdots,\al_s,\de_1,\cdots,\de_s\big)
=  A\big(\al_1+\de_1,\cdots,\al_s+\de_s\big)
- A\big(\al_1,\cdots,\al_s\big)
\end{equation*}
Then
\begin{equation*}
\tn\de A\tn_{w_\de}\le \sfrac{1}{\si} \tn A\tn_{w_\si}
\end{equation*}

\end{enumerate}
\end{lemma}

\begin{proof}
Let $a(\vec \bx_1,\cdots,\vec \bx_s)$ be a symmetric coefficient system 
for $f$. Since $a$ is invariant under permutation of its $\vec\bx_j$ components,
\begin{align*}
f\big(\al_1+\de_1,\cdots,\al_s+\de_s\big)
&=\sum_{(\vec\bx_1,\cdots,\vec\bx_s)\in\bX^{(s)}}
a(\vec\bx_1,\cdots,\vec\bx_s)\ 
        (\al_1+\de_1)(\vec\bx_1)\cdots(\al_s+\de_s)(\vec\bx_s)\\
&\hskip-25pt=\hskip-5pt\sum_{\atop{(\vec\bx_1,\cdots,\vec\bx_s)\in\bX^{(s)}}
                       {(\vec\by_1,\cdots,\vec\by_s)\in\bX^{(s)}} }\hskip-15pt
a(\vec\bx_1\circ\vec\by_1,\cdots,\vec\bx_s\circ\vec\by_s)\ 
            \prod_{j=1}^s\smchoose{n(\vec\bx_j)+n(\vec\by_j)}{n(\vec\by_j)}
             \al_j(\vec\bx_j)\de_j(\vec\by_j)
\end{align*}
so that
\begin{align*}
&\de a_p(\vec \bx_1,\!\cdots\!\bx_s ;\vec \by_1,\!\cdots\!,\by_s) \\
&\hskip1in=\chi\big(n(\vec\by_1)+\cdots+n(\vec\by_s)\ge p\big)
a(\vec\bx_1\circ\vec\by_1,\cdots,\vec\bx_s\circ\vec\by_s)\ 
            \prod_{j=1}^s\smchoose{n(\vec\bx_j)+n(\vec\by_j)}{n(\vec\by_j)}
\end{align*}
is a symmetric coefficient system for $\de f^{(\ge p)}$. Of course
$\de f=\de f^{(\ge 1)}$. By definition
\begin{align*}
\|\de f^{(\ge p)}\|_{w_\de}
&=\sum_{\atop{\atop{k_1,\cdots,k_s\ge 0}{\ell_1,\cdots,\ell_s\ge 0}}
             {\ell_1+\cdots+\ell_s\ge p}}
     \max_{\bx\in X}
    \max_{\atop{1\le j\le s}{k_j+\ell_j\ne 0}}
    \max_{1\le i\le k_j+\ell_j}
    \sum_{\atop{ \atop{\vec\bx_m\in X^{k_m}}{\vec\by_m\in X^{\ell_m}} }
               { {(\vec\bx_j\circ\vec\by_j)}_i=\bx } }
\big|\de a(\vec \bx_1,\!\cdots\!\bx_s ;\vec \by_1,\!\cdots\!,\by_s)\big|
\\\noalign{\vskip-0.4in}&\hskip3.4in
e^{\tau_d(\vec\bx_1\circ\vec\by_1,\cdots,\vec\bx_s\circ\vec\by_s)}
\prod_{j=1}^s\ka_j^{k_j}\la_j^{\ell_j}\\
&=\sum_{ \atop{ \atop{k_1,\cdots,k_s\ge 0}{\ell_1,\cdots,\ell_s\ge 0} }
              {\ell_1+\cdots+\ell_s\ge p} }
\om(k_1+\ell_1,\cdots,k_s+\ell_s)
\prod_{j=1}^s\smchoose{k_j+\ell_j}{\ell_j}
\prod_{j=1}^s\ka_j^{k_j}\la_j^{\ell_j}\\
&=\sum_{ \atop{n_1,\cdots,n_s}{n_1+\cdots+n_s\ge p} }
\om(n_1,\cdots,n_s)
c_p(n_1,\cdots,n_s)
\end{align*}
where
\begin{align*}
\om(n_1,\cdots,n_s)
&=\max_{\atop{1\le j\le s}{n_j\ne 0}}
  \max\limits_{1\le i\le n_j}
   \sum_{\atop{\vec\bz_p\in X^{n_p}}{ {(\vec\bz_j)}_i=\bx} }
\big|a(\vec\bz_1,\cdots,\vec\bz_s)\big|
e^{\tau_d(\vec\bz_1,\cdots,\vec\bz_s)}
\end{align*}
and
\begin{align*}
c_p(n_1,\cdots,n_s)
&=\sum_{ \atop{ \atop{k_j, \ell_j\ge 0}{k_j+\ell_j=n_j} }
              {\ell_1+\cdots+\ell_s\ge p} }
\prod_{j=1}^s\smchoose{n_j}{\ell_j}\ka_j^{k_j}\la_j^{\ell_j}
\le \sfrac{1}{\si^p}\prod_{j=1}^s (\ka_j+\si \la_j)^{n_j}
\end{align*}
For the last inequality, apply the
binomial expansion to each $(\ka_j+\si \la_j)^{n_j}$
and compare the two sides of the inequality  term by term.
This proves part (a). Part (b) follows by Remark \ref{remSUBsubtofn}.
\end{proof}

\begin{proposition}\label{propSUBsubstitution}
Let $h(\ga_1,\cdots,\ga_r)$ be an analytic function on a 
neighbourhood of the origin in $\bbbc^{r|X|}$, and let $A_j$, $\de A_j$,
$1\le j\le r$ be $s$--field maps. Furthermore let  
$\la_1,\ \cdots,\ \la_r$ 
be constant weight factors and let  $w_\la$ be the weight system with 
metric $d$ that associates the weight factor $\la_j$ to the field $\ga_j$. 

\begin{enumerate}[label=(\alph*), leftmargin=*]
\item  
Set
\begin{equation*}
\tilde h(\al_1,\cdots,\al_s) 
= h\Big( A_1(\al_1,\cdots,\al_s),\cdots, A_r(\al_1,\cdots,\al_s)\Big)
\end{equation*}
Assume that
\begin{equation*}
\tn A_j\tn_w \le \la_j
\end{equation*}
for each $1\le j\le r$. Then
\begin{equation*}
\|\tilde h\|_w \le \| h\|_{w_\la}
\end{equation*}

\item
Assume that there is a $\si\ge 1$ such that
\begin{equation*} 
\tn A_j\tn_w+\si\tn \de A_j\tn_w \le \la_j
\end{equation*}
for all $1\le j\le r$.
Set
\begin{align*}
& \widetilde{\de h}(\al_1,\cdots,\al_s)  \\
&\hskip0.2in= h\Big(
        \!A_1(\al_1,\cdots,\al_s\!)\!+\!\de A_1(\al_1,\cdots,\al_s),\cdots, 
         A_r(\al_1,\cdots,\al_s)\!+\!\de A_r(\al_1,\cdots,\al_s\!)\!\Big)\cr
&\hskip1in -h\Big( A_1(\al_1,\cdots,\al_s),\cdots, A_r(\al_1,\cdots,\al_s)
\Big)
\end{align*}
More generally, if $p\in\bbbn$ and $\de h^{(\ge p)}$ is the part of
\begin{equation*}
\de h(\ga_1,\cdots,\ga_r;\de_1,\cdots,\de_r)
=h(\ga_1+\de_1,\cdots,\ga_r+\de_r)-h(\ga_1,\cdots,\ga_r)
\end{equation*}
that is of degree at least $p$ in $(\de_1,\cdots,\de_r)$, set
\begin{align*}
&\widetilde{\de h}^{(\ge p)}(\al_1,\cdots,\al_s)\cr
&\hskip0.5in= \de h^{(\ge p)}\Big(
     A_1(\al_1,\cdots,\al_s),\cdots, A_r(\al_1,\cdots,\al_s)\,;\, \\
\noalign{\vskip-0.1in} &\hskip3in
   \de A_1(\al_1,\cdots,\al_s),\cdots, \de A_r(\al_1,\cdots,\al_s)
\Big)
\end{align*}
 Then
\begin{equation*}
\|\widetilde{\de h}\|_w \le \sfrac{1}{\si} \| h\|_{w_\la}\qquad
\big\|\widetilde{\de h}^{(\ge p)}\big\|_w \le \sfrac{1}{\si^p} \| h\|_{w_\la}
\end{equation*}
\end{enumerate}
\end{proposition}

\begin{proof}
(a) 
Let $a(\vec \by_1,\cdots,\vec \by_r)$ be a symmetric 
coefficient system for $h$. Define, for each $n(\vec\bx_i)=n_i\ge 0$, 
$1\le i\le s$,
\begin{align*}
&\tilde a(\vec \bx_1,\cdots,\vec \bx_s) \cr
& =\hskip-8pt\sum_{m_1,\cdots,m_r\ge 0} 
     \sum_{ \atop{  \atop{n_{i,j,k}\ge 0\ {\rm for}} 
                         {1\le i\le s,\,1\le j\le r,\,1\le k\le m_j} }
                 { {\rm with}\ \Si_{j,k}n_{i,j,k}=n_i }  
           }
    \sum_{  \atop{\vec\by_1\in X^{m_1}}
                 { \atop{\svdots}{\vec\by_r\in X^{m_r}}
                 } 
          }\hskip-9pt 
a(\vec \by_1,\cdots,\vec \by_r)
\prod_{j=1}^r\Big[\smprod_{k=1}^{m_j}
      A_j({(\vec\by_j)}_k; \vec \bx_{1,j,k},\cdots,\vec\bx_{s,j,k})\Big] 
\end{align*}
where ${(\vec\by_j)}_k$ is the $k^{\rm th}$ component of $\vec\by_j$ and
the $\vec\bx_{ijk}$'s are determined by the conditions that 
$n(\vec\bx_{ijk})=n_{ijk}$ and 
\begin{equation} \label{eqnSUBvecxidecomp}
\vec\bx_i
 =\circ_{j,k}\vec\bx_{ijk}
 =\vec\bx_{i11}\circ\vec\bx_{i12}\circ\cdots\circ\vec\bx_{i1m_1}
     \circ\vec\bx_{i21}\circ\cdots\circ\vec\bx_{i2m_2}\circ\cdots\circ
     \vec\bx_{irm_r}
\end{equation} 
Then 
$\tilde a(\vec \bx_1,\cdots,\vec \bx_s)$ is a (not necessarily symmetric)
coefficient system for $\tilde h$. Since
\begin{align*}
&\tau_d\big(\supp(\vec\bx_1,\cdots,\vec\bx_s)\big) \\
&\hskip0.5in\le \tau_d\big(\supp(\vec\by_1,\cdots,\vec\by_s)\big)+
\sum_{  \atop{1\le j\le r}{\le k\le m_j} }
    \tau_d\big(\supp({(\vec\by_j)}_k,
              \vec \bx_{1,j,k},\cdots,\vec\bx_{s,j,k})\big)
\end{align*}
we have
\begin{equation}\label{eqnSUBmessa}
\begin{split}
&w(\vec\bx_1,\cdots,\vec\bx_s)
\,   \big|\tilde a(\vec \bx_1,\cdots,\vec \bx_s)\big|\\
&\hskip0.5in   
\le\sum_{m_1,\cdots,m_r\ge 0}\ 
 \sum_{ \atop{n_{i,j,k}\ge 0\ {\rm for}} 
             { \atop{1\le i\le s,\,1\le j\le r,\,1\le k\le m_j}
                    { {\rm with}\ \Si_{j,k}n_{i,j,k}=n_i} 
             }
      }\ 
 \sum_{ \atop{\vec\by_1\in X^{m_1}}
             { \atop{\svdots}{\vec\by_r\in X^{m_r}}  }
      }\hskip-6pt 
 w_\la(\vec\by_1,\cdots,\vec\by_r)
\big|a(\vec \by_1,\cdots,\vec \by_r)\big|\\
& \hskip 2.8in
\prod_{j=1}^r\Big[\smprod_{k=1}^{m_j}
       B_j({(\vec\by_j)}_k; \vec \bx_{1,j,k},\cdots,\vec\bx_{s,j,k}) \Big]
\end{split}
\end{equation}
where
\begin{equation*}
B_j(\by; \vec \bx_1',\cdots,\vec\bx_s')
=\sfrac{1}{\la_j}
   |A_j(\by; \vec \bx_1',\cdots,\vec\bx_s')|
         \ka_1^{n(\vec\bx_1')}\cdots \ka_s^{n(\vec\bx_s')} 
    e^{\tau_d(\supp(\by, \vec \bx_1',\cdots,\vec\bx_s'))}
\end{equation*}

We first observe that when $\vec\bx_1=\cdots=\vec\bx_s=-$, we have
$
\tilde a(-,\cdots,-)
= a(-,\cdots,-)
$
so that the corresponding contributions to $\|\tilde h\|_{w}$
and $\|h\|_{w_\la}$ are identical. Therefore we may assume, without loss of 
generality, that $h(0,\cdots,0)=0$.

We are to bound
\begin{equation*}
\|\tilde h\|_{w}
=\hskip-3pt\sum_{\atop{n_1,\cdots,n_{s}\ge 0}{n_1+\cdots+n_s\ge 1}}\hskip-2pt
\max_{\bx\in X} \max_{\atop{1\le \bar\jmath\le s}{n_{\bar\jmath}\ne 0}}
\max_{1\le\bar\imath\le n_{\bar\jmath}}
\sum_{\atop{(\vec\bx_1,\cdots,\vec\bx_s)\in X^{n_1}\times\cdots\times X^{n_s}}
           {{(\vec\bx_{\bar\jmath})}_{\bar\imath}=\bx} }
\hskip-30pt
 w(\vec\bx_1,\cdots,\vec\bx_s)\ 
\big|\tilde a(\vec\bx_1,\cdots,\vec\bx_s)\big|
\end{equation*}
First fix any $n_1,\ \cdots,\ n_s\ge 0$ with $n_1+\cdots+n_s\ge 1$. We claim that
\begin{equation}\label{eqnSUBmaxs}
\begin{split}
& \max_{\bx\in X}
  \max_{\atop{1\le \bar\jmath\le s}{n_{\bar\jmath}\ne 0}}
  \max_{1\le\bar\imath\le n_{\bar\jmath}}
  \sum_{\atop{ (\vec\bx_1,\cdots,\vec\bx_s)\in X^{n_1}\times\cdots\times X^{n_s} }
             {{(\vec\bx_{\bar\jmath})}_{\bar\imath}=\bx}}
\hskip-30pt
 w(\vec\bx_1,\cdots,\vec\bx_s)\ 
\big|\tilde a(\vec\bx_1,\cdots,\vec\bx_s)\big|\\
&\hskip0.5in\le \sum_{m_1,\cdots,m_r\ge 0} \|w_\la a\big\|_{m_1,\cdots,m_r}
 \sum_{  \atop{n_{i,j,k}\ge 0\ {\rm for}}
              {\atop{1\le i\le s,\,1\le j\le r,\,1\le k\le m_j}
                    {{\rm with}\ \Si_{j,k}n_{i,j,k}=n_i}}
       }
\prod_{\atop{1\le j\le r}{1\le k\le m_j}}\Big[
   \sfrac{1}{\la_j} \big\|A_j\big\|_{w;n_{1,j,k},\cdots,n_{s,j,k}} \Big]
\end{split}
\end{equation}
Here,
as in \cite[Definition \defCPCnorms]{CPC},
\begin{equation*}
\|b\|_{m_1,\cdots,m_r}
=\max_{\by\in X}
 \max_{\atop{1\le j\le r}{m_j\ne 0}}
 \max\limits_{1\le i\le m_j}
 \sum_{ \atop{ \atop{\vec\by_\ell\in X^{m_\ell}}{1\le \ell\le r} }
             {{(\vec\by_j)}_i=\by}
      }
\big|b(\vec\by_1,\cdots,\vec\by_r)\big|
\end{equation*}
To prove \eqref{eqnSUBmaxs}, fix any $\bx\in X$ and assume, without loss of 
generality that $n_1\ge 1$ and $\bar\jmath=\bar\imath=1$. By \eqref{eqnSUBmessa},
(the meaning of the $\hat\jmath,\hat k$ introduced after the ``$=$'' below
is explained immediately following this string of inequalities) 
\begin{align*}
&\sum_{\atop{(\vec\bx_1,\cdots,\vec\bx_s)\in X^{n_1}\times\cdots\times X^{n_s}}
            {{(\vec\bx_1)}_1=\bx}
      }
\hskip-30pt w(\vec\bx_1,\cdots,\vec\bx_s)\ 
\big|\tilde a(\vec\bx_1,\cdots,\vec\bx_s)\big|\\
&\le \hskip-10pt
 \sum_{ \atop{(\vec\bx_1,\cdots,\vec\bx_s)\in X^{n_1}\times\cdots\times X^{n_s}}
             {{(\vec\bx_1)}_1=\bx}
       } 
 \sum_{m_1,\cdots,m_r\ge 0}\ 
 \sum_{ \atop{ \atop{n_{i,j,k}\ge 0\ {\rm for}} 
                    {1\le i\le s,\,1\le j\le r,\,1\le k\le m_j} }
              {{\rm with}\ \Si_{j,k}n_{i,j,k}=n_i }
       }\ 
  \sum_{ \atop{\vec\by_1\in X^{m_1}}
              { \atop{\svdots}{\vec\by_r\in X^{m_r}}
              }
       }\hskip-6pt 
  w_\la(\vec\by_1,\cdots,\vec\by_r)\big|a(\vec \by_1,\cdots,\vec \by_r)\big|\\
& \hskip 3in
  \prod_{j=1}^r\Big[\smprod_{k=1}^{m_j}
       B_j({(\vec\by_j)}_k; \vec \bx_{1,j,k},\cdots,\vec\bx_{s,j,k}) \Big]
\displaybreak[0]\\
&= \sum_{m_1,\cdots,m_r\ge 0}
  \sum_{ \atop{n_{i,j,k}\ge 0\ {\rm for}} 
              { \atop{1\le i\le s,\,1\le j\le r,\,1\le k\le m_j}
                     {{\rm with}\ \Si_{j,k}n_{i,j,k}=n_i } }
        }\ 
  \sum_{  \atop{\vec\bx_{i,j,k}\in X^{n_{i,j,k}}\ {\rm for}} 
               {  \atop{1\le i\le s,\,1\le j\le r,\,1\le k\le m_j}
                       {{\rm with}\ {(\vec\bx_{1,\hat\jmath,\hat k})}_1=\bx }
               }
        }\ 
\sum_{  \atop{\vec\by_1\in X^{m_1}}
             { \atop{\svdots}
                    {\vec\by_r\in X^{m_r}}
             }
      }\hskip-6pt 
      w_\la(\vec\by_1,\cdots,\vec\by_r)\big|a(\vec \by_1,\cdots,\vec \by_r)\big|\cr
 \noalign{\vskip-0.05in} & \hskip 3in
\prod_{j=1}^r\Big[\smprod_{k=1}^{m_j}
       B_j({(\vec\by_j)}_k; \vec \bx_{1,j,k},\cdots,\vec\bx_{s,j,k}) \Big]
\displaybreak[0]\\
&\le 
\sum_{m_1,\cdots,m_r\ge 0}
 \sum_{  \atop{\atop{n_{i,j,k}\ge 0\ {\rm for}} 
                    {1\le i\le s,\,1\le j\le r,\,1\le k\le m_j} }
              {{\rm with}\ \Si_{j,k}n_{i,j,k}=n_i}
      }\ 
 \sum_{ \atop{\atop{\vec\bx_{i,\hat\jmath,\hat k}\in X^{n_{i,\hat\jmath,\hat k}}\ {\rm for}}
                   {1\le i\le s} }
             {{\rm with}\ {(\vec\bx_{1,\hat\jmath,\hat k})}_1=\bx}
      }\ 
\sum_{  \atop{\vec\by_1\in X^{m_1}}
             {\atop{\svdots} {\vec\by_r\in X^{m_r}}}
     }\hskip-6pt 
 w_\la(\vec\by_1,\cdots,\vec\by_r)
\big|a(\vec \by_1,\cdots,\vec \by_r)\big|\\
 \noalign{\vskip-0.05in} & \hskip1.2in
  B_{\hat\jmath}({(\vec\by_{\hat\jmath})}_{\hat k}; 
      \vec \bx_{1,\hat\jmath,\hat k},\cdots,\vec\bx_{s,\hat\jmath,\hat k}) 
\prod_{  \atop{\atop{1\le j\le r}{1\le k\le m_j}}
              {(j,k)\ne(\hat\jmath,\hat k)}
      }\Big[
    \sfrac{1}{\la_j} L\big(A_j;w;\big\{n_{i,j,k}\big\}_{1\le i\le s}\big) \Big]
\displaybreak[0]\\
&\le 
\sum_{m_1,\cdots,m_r\ge 0}
 \sum_{ \atop{\atop{n_{i,j,k}\ge 0\ {\rm for}} 
                   {1\le i\le s,\,1\le j\le r,\,1\le k\le m_j} }
             {{\rm with}\ \Si_{j,k}n_{i,j,k}=n_i}
      }\ 
 \sum_{  \atop{\atop{\vec\bx_{i,\hat\jmath,\hat k}\in X^{n_{i,\hat\jmath,\hat k}}\ {\rm for}}
                    {1\le i\le s} }
              {{\rm with}\ {(\vec\bx_{1,\hat\jmath,\hat k})}_1=\bx}
      }\ 
   \sum_{\by\in X} \|w_\la a\big\|_{m_1,\cdots,m_r}\\
 \noalign{\vskip-0.07in}& \hskip1.2in
     B_{\hat\jmath}(\by; \vec \bx_{1,\hat\jmath,\hat k},\cdots,
                                       \vec\bx_{s,\hat\jmath,\hat k}) 
    \prod_{ \atop{\atop{1\le j\le r} {1\le k\le m_j }}
                 {(j,k)\ne(\hat\jmath,\hat k)}
          }\Big[
    \sfrac{1}{\la_j} L\big(A_j;w;\big\{n_{i,j,k}\big\}_{1\le i\le s}\big) \Big]
\displaybreak[0]\\
&\le 
\sum_{m_1,\cdots,m_r\ge 0} \|w_\la a\big\|_{m_1,\cdots,m_r}
 \sum_{ \atop{\atop {n_{i,j,k}\ge 0\ {\rm for}} 
                    {1\le i\le s,\,1\le j\le r,\,1\le k\le m_j} }
             {{\rm with}\ \Si_{j,k}n_{i,j,k}=n_i}
      }\hskip-15pt
       \sfrac{1}{\la_{\hat\jmath}}R\big(A_{\hat\jmath};w;
               \big\{n_{i,\hat\jmath,\hat k}\big\}_{1\le i\le s}\big)\\
 \noalign{\vskip-0.05in}&\hskip2.5in
 \prod_{  \atop{\atop{1\le j\le r}{1\le k\le m_j} }
               {(j,k)\ne(\hat\jmath,\hat k)}
       }\Big[
   \sfrac{1}{\la_j} L\big(A_j;w;\big\{n_{i,j,k}\big\}_{1\le i\le s}\big) \Big]
\displaybreak[0]\\
&\le 
\sum_{m_1,\cdots,m_r\ge 0} \|w_\la a\big\|_{m_1,\cdots,m_r}
 \sum_{ \atop{\atop{n_{i,j,k}\ge 0\ {\rm for}} 
                   {1\le i\le s,\,1\le j\le r,\,1\le k\le m_j}}
             {{\rm with}\ \Si_{j,k}n_{i,j,k}=n_i}
       }
\prod_{ \atop{1\le j\le r}{1\le k\le m_j}}\Big[
   \sfrac{1}{\la_j} \big\|A_j\big\|_{w;n_{1,j,k},\cdots,n_{s,j,k}} \Big]
\end{align*}
Here, for each $\big\{n_{1,j,k}\big\}_{ \atop{1\le j\le r}{1\le k\le m_j} }$,
the pair $(\hat\jmath,\hat k)$ is the first $(j,k)$, using the lexicographical
ordering of \eqref{eqnSUBvecxidecomp}, for which $n_{1,j,k}\ne 0$.

\begin{center}
   \includegraphics{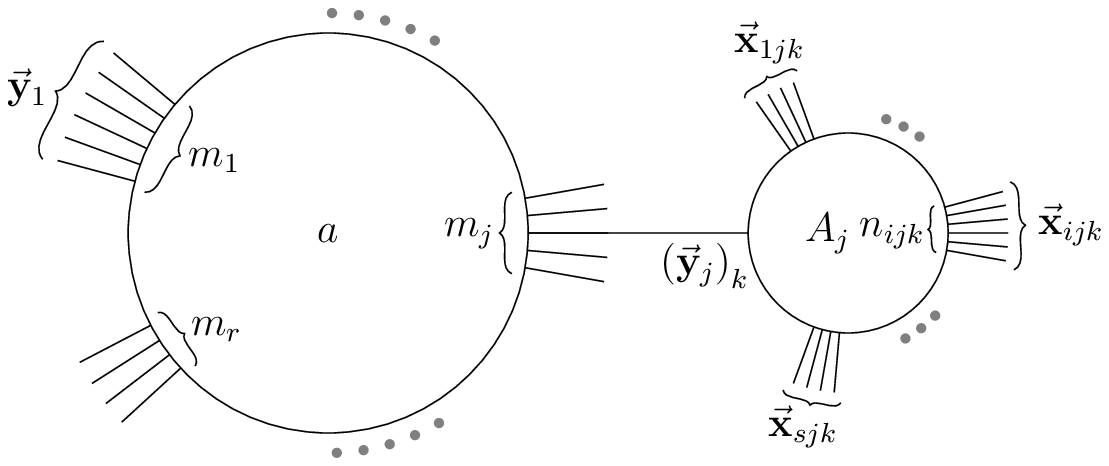}
\end{center}

Having completed the proof of \eqref{eqnSUBmaxs}, we now have,
recalling the hypothesis that each $\tn A_j\tn_w\le\la_j$,
\begin{align*}
\|\tilde h\|_{w}
&\le\hskip-3pt\sum_{ \atop{n_1,\cdots,n_{s}\ge 0}{n_1+\cdots+n_s\ge 1}}
\sum_{m_1,\cdots,m_r\ge 0} \|w_\la a\big\|_{m_1,\cdots,m_r}\hskip-20pt
 \sum_{  \atop{\atop{n_{i,j,k}\ge 0\ {\rm for}} 
                    {1\le i\le s,\,1\le j\le r,\,1\le k\le m_j}}
              {{\rm with}\ \Si_{j,k}n_{i,j,k}=n_i}
      }
\prod_{  \atop{1\le j\le r}{1\le k\le m_j}}\Big[
   \sfrac{1}{\la_j} \big\|A_j\big\|_{w;n_{1,j,k},\cdots,n_{s,j,k}} \Big]\cr
&\le
\sum_{m_1,\cdots,m_r\ge 0} \|w_\la a\big\|_{m_1,\cdots,m_r}\hskip-10pt
 \sum_{ \atop{n_{i,j,k}\ge 0\ {\rm for}}
             {1\le i\le s,\,1\le j\le r,\,1\le k\le m_j}
      }
\prod_{{1\le j\le r\atop 1\le k\le m_j}}\Big[
   \sfrac{1}{\la_j} \big\|A_j\big\|_{w;n_{1,j,k},\cdots,n_{s,j,k}} \Big]\cr
&\le
\sum_{m_1,\cdots,m_r\ge 0} \|w_\la a\big\|_{m_1,\cdots,m_r}\cr
&=\|h\|_{w_\la}
\end{align*}

\Item (b)
Let $w_\de$ be the weight system with metric $d$ that associates 
the weight factor $\tn A_j\tn_w$ to $\ga_j$ and the 
weight factor $\tn \de A_j\tn_w$ to $\de_j$. By part (a) of Lemma
\ref{lemSUBdiff}, with 
\begin{align*}
&f\rightarrow h  \\
&s\rightarrow r  \\
&\al_j\text{ with weight }\ka_j\rightarrow 
         \ga_j\text{ with weight }\tn A_j\tn_w \\
&\de_j\text{ with weight }\la_j\rightarrow 
         \de_j\text{ with weight }\tn\de A_j\tn_w 
\end{align*}
we have
\begin{equation*}
\|\de h\|_{w_\de}\le\sfrac{1}{\si}\|h\|_{w_\la}\qquad
\|\de h^{(\ge p)}\|_{w_\de}\le\sfrac{1}{\si^p}\|h\|_{w_\la}
\end{equation*}
Now $\widetilde{\de h}$ and $\widetilde{\de h}^{(\ge p)}$ are obtained 
from $\de h$ and $\de h^{(\ge p)}$, respectively, by the substitutions
\begin{equation*}
\ga_j= A_j(\al_1,\cdots,\al_s)\qquad 
\de_j= \de A_j(\al_1,\cdots,\al_s)
\end{equation*}
and the statement follows by part (a).
\end{proof}

\begin{corollary}\label{corSUBsubstitution}
Let $B$ be an $r$--field map and let $A_j$, $1\le j\le r$, 
be $s$--field maps. Define the $s$--field map $\tilde B$ by
\begin{equation*}
\tilde B(\al_1,\cdots,\al_s) 
= B\Big( A_1(\al_1,\cdots,\al_s),\cdots, A_r(\al_1,\cdots,\al_s)\Big)
\end{equation*}
Furthermore let  $\la_1,\ \cdots,\ \la_r$ be constant weight factors and
let  $w_\la$ be the weight system with metric $d$ that associates 
the weight factor $\la_j$ to the $j^{\rm th}$ field of $B$. Assume that
\begin{equation*}
\tn A_j\tn_w \le \la_j
\end{equation*}
for each $1\le j\le r$. Then
\begin{equation*}
\tn\tilde B\tn_w \le \tn B\tn_{w_\la}
\end{equation*}
\end{corollary}

\begin{proof}
This follows from Proposition \ref{propSUBsubstitution}
and Remark \ref{remSUBsubtofn}.
\end{proof}

\begin{definition}\label{defSUBcB}
Denote by $w_{\ka,\la}$ the weight system with metric $d$ that 
associates the constant weight factor $\ka_i$ to the field 
$\al_i$ and the constant weight factor $\la_j$ to the field $\ga_j$.
Let $B(\vec\al,\vec\ga)$ be an $(s+r)$--field map with 
$\tn B\tn_{w_{\ka,\la}}<\infty$.

\begin{enumerate}[label=(\alph*), leftmargin=*]
\item
Set, for each r--tuple of nonnegative integers
$n_{s+1}$, $\cdots$, $n_{s+r}$,
\begin{align*}
&B_{n_{s+1},\cdots,n_{s+r}}(\bx;\vec\bx_1,\cdots,\vec\bx_{s+r}) \\
&\hskip1in=\begin{cases}
        B(\bx;\vec\bx_1,\cdots,\vec\bx_{s+r})
             &\text{if $n(\vec\bx_{s+j})=n_{s+j}$ for all $1\le j\le r$}\\
        0    &\text{otherwise}
 \end{cases}
\end{align*}
Then 
\begin{equation*}
B=\sum_{n_{s+1},\cdots,n_{s+r}\ge 0} \!\! B_{n_{s+1},\cdots,n_{s+r}}
\quad
\text{and}
\quad
\tn B\tn_{w_{\ka,\la}} = 
\sum_{n_{s+1},\cdots,n_{s+r}\ge 0} \!\!
         \tn B_{n_{s+1},\cdots,n_{s+r}} \tn_{w_{\ka,\la}}
\end{equation*}
$B$ is said to have minimum degree
at least $d_{\rm min}$ and maximum degree at most $d_{\rm max}\le\infty$ in its last
$r$ arguments if
\begin{equation*}
B_{n_{s+1},\cdots,n_{s+r}}=0\quad
\text{unless}\quad 
d_{\rm min}\le n_{s+1} + \cdots  + n_{s+r} \le d_{\rm max}
\end{equation*}
Set
\begin{equation*}
\tn B\tn'_{w_{\ka,\la}} = 
\sum_{n_{s+1},\cdots,n_{s+r}\ge 0}
         \big(n_{s+1} + \cdots  + n_{s+r}\big)
         \tn B_{n_{s+1},\cdots,n_{s+r}} \tn_{w_{\ka,\la}}
\end{equation*}
Think of $\tn B\tn'_{w_{\ka,\la}}$ as a bound on the derivative of
$B(\vec\al,\vec\ga)$ with respect to $\vec\ga$. See Lemma \ref{lemSUBdiffnorm}.

\item
Denote by $\cB$ the Banach space of all $r$--tuples $\vec\Ga=(\Ga_1,\cdots,\Ga_r)$
of $s$--field maps with the norm 
\begin{equation*}
\|\vec\Ga\|=\max_{1\le j\le r}\sfrac{1}{\la_j}\tn\Ga_j\tn_w
\end{equation*}
Also, for each $\rho>0$, denote by $\cB_\rho$, the closed ball in $\cB$ of radius
$\rho$.

\item
For each  $r$--tuple $\vec\Ga\in\cB_1$, we define  
the $s$--field map $\tilde B(\vec\Ga)$ by
\begin{equation*}
\big(\tilde B(\vec\Ga)\big)(\vec\al)
                  =B\big(\vec\al,\vec\Ga(\vec\al)\big)
\end{equation*}
\end{enumerate}
\end{definition}

\begin{remark}\label{remSUBcB}
Let  $B$ be an $(s+r)$--field map with  minimum degree
at least $d_{\rm min}$ and maximum degree at most $d_{\rm max}<\infty$ in its last
$r$ arguments. 

\begin{enumerate}[label=(\alph*), leftmargin=*]
\item  
$\ d_{\rm min} \tn B\tn_{w_{\ka,\la}}
               \le \tn B\tn'_{w_{\ka,\la}} 
               \le  d_{\rm max} \tn B\tn_{w_{\ka,\la}} $

\item 
If $d_{\rm min} = d_{\rm max}=1$,
$B$ is said to be linear. In this case, for any fixed $\al_1$, $\cdots$, 
$\al_s$, the map
\begin{equation*}
(\ga_1,\cdots,\ga_s)\mapsto 
B(\al_1,\cdots,\al_s,\ga_1,\cdots\ga_r)
\end{equation*}
is linear and $\tn B\tn'_{w_{\ka,\la}}= \tn B\tn_{w_{\ka,\la}} $
\end{enumerate}
\end{remark}

\begin{example}\label{egSUBcB}
A simple example with $s=0$ and $r=1$ is the truncated exponential
\begin{equation*}
B\big(\ga\big)(\bx) = E_n\big(a\ga(\bx)\big)
\qquad\text{where}\qquad
E_n(z)=\sum_{\ell=n}^\infty\sfrac{1}{\ell!} z^\ell 
\end{equation*}
and $a$ is a constant.
In this example, $B$ is a local function of $\ga$, so that all of the kernels
of $B$ are just delta functions. Hence
\begin{align*}
 \tn B\tn_{w_{\ka,\la}}
  & =  \sum_{\ell=n}^\infty\sfrac{1}{\ell!} a^\ell \la^\ell
     = E_n(a\la)
     \le \sfrac{a^n\la^n}{n!} e^{a\la} \cr
 \tn B\tn'_{w_{\ka,\la}}
  & =  \sum_{\ell=n}^\infty\sfrac{1}{(\ell-1)!} a^\ell \la^\ell
     = a\la E_{n-1} (a\la)
     \le \sfrac{a^n\la^n}{(n-1)!} e^{a\la} \cr
\end{align*}
\end{example}

\begin{lemma}\label{lemSUBdiffnorm}
Let $B$ be an $(s+r)$--field map with $\tn B\tn'_{w_{\ka,\la}}<\infty$.
Assume that $B$ has minimum degree at least $d_{\rm min}$
in its last $r$ arguments.
Then, for each $\vec\Ga,\vec\Ga'\in\cB_1$,
\begin{equation*}
\tn\tilde B(\vec\Ga)-\tilde B(\vec\Ga')\tn_w
\le \|\vec\Ga-\vec\Ga'\|\, 
    \max\big\{\|\vec\Ga\|\,,\,\|\vec\Ga'\| \big\}^{d_{\rm min}-1}
     \tn B\tn'_{w_{\ka,\la}}
\end{equation*}
\end{lemma}

\begin{proof} 
Write
\begin{equation*}
B=\sum_{\atop{n_{s+1},\cdots,n_{s+r}\ge 0}
             {n_{s+1}+\cdots+n_{s+r}\ge d_{\rm min}}
        }
         B_{n_{s+1},\cdots,n_{s+r}}
\end{equation*}
as in Definition  \ref{defSUBcB}.
Since
\begin{align*}
\tn B\tn_{w_{\ka,\la}} &= \hskip-10pt
\sum_{\atop{n_{s+1},\cdots,n_{s+r}\ge 0}
           {n_{s+1}+\cdots+n_{s+r}\ge d_{\rm min}}
      }\hskip-20pt
         \tn B_{n_{s+1},\cdots,n_{s+r}} \tn_{w_{\ka,\la}}
\\
\tn B\tn'_{w_{\ka,\la}} &= \hskip-10pt
\sum_{ \atop{n_{s+1},\cdots,n_{s+r}\ge 0}
            {n_{s+1}+\cdots+n_{s+r}\ge d_{\rm min}}
      }\hskip-20pt
         \tn B_{n_{s+1},\cdots,n_{s+r}} \tn'_{w_{\ka,\la}}
\end{align*}
we may assume, without loss of generality, that at most one 
$B_{n_{s+1},\cdots,n_{s+r}}$ is nonvanishing. By renaming the $\ga$ 
fields and changing the value of $r$, we may assume that 
$n_{s+1} =\cdots = n_{s+r}=1$. Then 
$B\big(\vec\al,\ga_1,\cdots\ga_r)$ is multilinear in $\ga_1$, $\cdots$,
$\ga_r$ so that
\begin{align*}
&\tilde B(\vec\Ga)(\vec\al)-\tilde B(\vec\Ga')(\vec\al)
=B\big(\vec\al,\Ga_1(\vec\al),\cdots,\Ga_r(\vec\al)\big)
  -B\big(\vec\al,\Ga'_1(\vec\al),\cdots,\Ga'_r(\vec\al)\big)\\
&\hskip1in=\sum_{j=1}^r
  B\Big(\vec\al\,,\,\Ga_1(\vec\al),\cdots,\Ga_{j-1}(\vec\al),
               \Ga_j(\vec\al)-\Ga'_j(\vec\al),
               \Ga'_{j+1}(\vec\al),\cdots, \Ga'_r(\vec\al)\Big)
\end{align*}
So, by Corollary \ref{corSUBsubstitution},
\begin{align*}
\tn \tilde B(\vec\Ga)-\tilde B(\vec\Ga')\tn_w
&\le \sum_{j=1}^r
  \Big(\prod_{k=1}^{j-1}\sfrac{\tn \Ga_k\tn_w}{\la_k}\Big)
  \sfrac{\tn \Ga_j-\Ga'_j\tn_w}{\la_j}
  \Big(\prod_{k=j+1}^{r}\sfrac{\tn \Ga'_k\tn_w}{\la_k}\Big)
  \tn B\tn_{w_{\ka,\la}}\\
&\le r\ 
  \max\big\{\|\vec\Ga\|\,,\,\|\vec\Ga'\| \big\}^{r-1}\ 
  \|\vec\Ga-\vec\Ga'\|\ \tn B\tn_{w_{\ka,\la}}\\
&\le 
  \max\big\{\|\vec\Ga\|\,,\,\|\vec\Ga'\| \big\}^{r-1}\ 
  \|\vec\Ga-\vec\Ga'\|\ \tn B\tn'_{w_{\ka,\la}}
\end{align*}
The claim follows since $\max\big\{\|\vec\Ga\|\,,\,\|\vec\Ga'\| \big\}\le 1$
and $r\ge d_{\rm min}$.
\end{proof}

\begin{lemma}[Product Rule]\label{lemSUBprodRule}
Let $A(\vec\al,\vec\ga)$ and $B(\vec\al,\vec\ga)$ be $(s+r)$--field maps 
with $\tn A\tn'_{w_{\ka,\la}}, \tn B\tn'_{w_{\ka,\la}}<\infty$.
Define 
\begin{equation*}
C(\vec\al,\vec\ga)(\bx) = A(\vec\al,\vec\ga)(\bx)\ B(\vec\al,\vec\ga)(\bx)
\end{equation*}
Then
\begin{equation*}
\tn C\tn'_{w_{\ka,\la}}
 \le    \tn A\tn'_{w_{\ka,\la}}\tn B\tn_{w_{\ka,\la}}
        + \tn A\tn_{w_{\ka,\la}}\tn B\tn'_{w_{\ka,\la}}
\end{equation*}
\end{lemma}
\begin{proof} 
For convenience of notation, write $\vec n=(n_{s+1},\cdots,n_{s+r})$,
$|\vec n| = n_{s+1} + \cdots  + n_{s+r}$
and $\vec n\ge 0$ for $n_{s+1},\cdots,n_{s+r}\ge 0$. Then, in
the notation of Definition \ref{defSUBcB}.a,
\begin{equation*}
C=\sum_{\vec N\ge 0} C_{\vec N}\qquad\hbox{with}\qquad
C_{\vec N} = \sum_{\ \atop{vec n,\vec m\ge 0}{\vec n+\vec m=\vec N}}
                    A_{\vec n} B_{\vec m}
\end{equation*}
and
\begin{align*}
\tn C\tn'_{w_{\ka,\la}}
&= \sum_{\vec N\ge 0} |\vec N|\, \tn C_{\vec N}\tn_{w_{\ka,\la}} \\
&\le \sum_{\vec n,\vec m\ge 0} \big(|\vec n| + |\vec m|\big)\, 
          \tn A_{\vec n} B_{\vec m} \tn_{w_{\ka,\la}}
\end{align*}
So the claim follows from
\begin{equation*}
\tn A_{\vec n} B_{\vec m} \tn_{w_{\ka,\la}}
   \le \tn A_{\vec n} \tn_{w_{\ka,\la}}
       \tn B_{\vec m} \tn_{w_{\ka,\la}}
\end{equation*}
\end{proof}

\newpage
\section{Solving Equations}\label{secSUBsolve}
In this section we consider  systems of $r\ge 1$ implicit equations of the form
\refstepcounter{equation}\label{eqnSUBfixedpteqn}
\begin{equation}
\ga_j=f_j(\vec\al)+L_j(\vec\al,\vec\ga)+B_j\big(\vec\al,\vec\ga\big)
\tag{\ref{eqnSUBfixedpteqn}.a}
\end{equation}
for  ``unknown''  fields $\ga_1,\cdots,\ga_r$ as a function of fields $\al_1,\cdots,\al_s$.
In the above equation, $\vec\al=\big(\al_1,\cdots,\al_s\big)$, $\vec\ga=\big(\ga_1,\cdots,\ga_r\big)$, and
for each $1\le j\le r$,
\begin{itemize}[leftmargin=*, topsep=2pt, itemsep=0pt, parsep=0pt]
\item
   $f_j$ is an $s$--field map,

\item 
    $L_j$ is an $(s+r)$--field map that is linear in
    its last $r$ arguments, and

\item
     $B_j$ is an $(s+r)$--field map.
\end{itemize}
We write the system  (\ref{eqnSUBfixedpteqn}.a) in the shorthand notation
\begin{equation}
\vec \ga
=\vec f(\vec\al)+\vec L(\vec\al,\vec\ga)+\vec B\big(\vec\al,\vec\ga\big)
\tag{\ref{eqnSUBfixedpteqn}.b}
\end{equation}

Example \ref{exSUBbackgroundfield}, below, is of this form and is a 
simplified version of the kind of equations that occur
as equations for ``background fields'' and ``critical fields'' in 
\cite{PAR1,PAR2}.
The following proposition gives conditions under which this system of 
equations has a solution
 $\vec \ga = \vec\Ga(\vec\al)$,
estimates on the solution, and a uniqueness statement.

\begin{proposition}\label{propSUBeqnsoln}
Let $\ka_1$, $\cdots$, $\ka_s$ and $\la_1$, $\cdots$, $\la_r$ 
 be constant weight factors for the fields $\al_1,\cdots,\al_s$
and $\ga_1,\cdots,\ga_r$, respectively.
As in  Definition \ref{defSUBcB} set $\cB_1=\set{\vec\Ga}{\|\vec\Ga\|\le 1}$ where 
$
\|\vec\Ga\|=\max\limits_{1\le j\le r}\sfrac{1}{\la_j}\tn\Ga_j\tn_{w_\ka}
$. 
Let  $0<\fc<1$ be a contraction factor.

\noindent
Assume that, for each $1\le j\le r$,
the $(s+r)$--field map $B_j(\vec\al;\vec\ga)$ has
minimum degree at least $2$  in its last
$r$ arguments (that is, in $\vec\ga$). Also  assume that for $1\le j \le r$
\begin{align*}
\tn f_j\tn_{w_\ka}+\tn L_j\tn_{w_{\ka,\la}}+\tn B_j\tn_{w_{\ka,\la}}&\le\la_j\\
 \tn L_j\tn_{w_{\ka,\la}}+\tn B_j\tn'_{w_{\ka,\la}}&\le \fc\la_j
\end{align*}
\begin{enumerate}[label=(\alph*), leftmargin=*]
\item
Then there is a unique $\vec\Ga\in\cB_1$ for which
\begin{equation*} 
\vec\Ga(\vec\al)
=\vec f(\vec\al)+\vec L\big(\vec\al,\vec\Ga(\vec\al)\big)
   +\vec B\big(\vec\al,\vec\Ga(\vec\al)\big)
\end{equation*}
That is, which solves \eqref{eqnSUBfixedpteqn}. Furthermore
\begin{equation*}
\max_{j}\sfrac{1}{\la_j}\tn\Ga_j\tn_w
         \le\sfrac{1}{1-\fc}\max_{j}\sfrac{1}{\la_j}\tn f_j\tn_w\qquad
\max_{j}\sfrac{1}{\la_j}\tn\Ga_j-f_j\tn_w
         \le\sfrac{\fc}{1-\fc}\max_{j}\sfrac{1}{\la_j}\tn f_j\tn_w
\end{equation*}

\item
Assume, in addition, that
\begin{equation*}
\tn f_j\tn_w\le (1-\fc)^2\,\la_j\qquad\hbox{for all }1\le j\le r
\end{equation*}
Denote by $\vec\Ga$ the solution of part (a) and by
 $\vec\Ga^{(1)}$ the unique element of $\cB_1$ that solves
$\ga_j=f_j(\vec\al)+L_j(\vec\al,\vec\ga)$ for $1\le j\le r$. Then
\begin{equation*}
\|\vec\Ga^{(1)}\|\le\sfrac{1}{1-\fc}\|\vec f\,\|\qquad
\|\vec\Ga^{(1)}-\vec f\|\le\sfrac{\fc}{1-\fc}\|\vec f\,\|
\end{equation*}
and
\begin{equation*}
\|\vec\Ga-\vec\Ga^{(1)}\|
      \le \sfrac{\|\vec f\,\|^2}{(1-\fc)^3} \max_{1\le j\le r}\sfrac{1}{\la_j}\tn B_j\tn_{w_{\ka,\la}}
      \le  \max_{1\le j\le r}\sfrac{1}{\la_j}\tn B_j\tn_{w_{\ka,\la}}
\end{equation*}
\end{enumerate}
\end{proposition}

\begin{proof}
(a)
 Define $F(\vec\Ga)$ by
\begin{equation*}
\vec F(\vec\Ga)=\left[\begin{matrix}
                           f_1+\tilde L_1(\vec\Ga)+\tilde B_1(\vec\Ga)\\
                           \vdots\\
                           \noalign{\vskip0.05in}
                         f_r+\tilde L_r(\vec\Ga)+\tilde B_r(\vec\Ga)  
                      \end{matrix}\right]
\end{equation*}
Recall, from Definition \ref{defSUBcB}, that 
\begin{equation*}
   \big(\tilde L_j(\vec\Ga)\big)(\vec\al)
                  =L_j\big(\vec\al,\vec\Ga(\vec\al)\big)
   \quad\text{and}\quad
    \big(\tilde B_j(\vec\Ga)\big)(\vec\al)
                  =B_j\big(\vec\al,\vec\Ga(\vec\al)\big)
\end{equation*}
By Corollary \ref{corSUBsubstitution} and the hypothesis
$\tn f_j\tn_{w_\ka}+\tn L_j\tn_{w_{\ka,\la}}+\tn B_j\tn_{w_{\ka,\la}}\le\la_j$, 
$\vec F$ maps $\cB_1$ into $\cB_1$.
By Lemma \ref{lemSUBdiffnorm} and Remark \ref{remSUBcB}.b,
$\|\vec F(\vec\Ga)-\vec F(\vec\Ga')\|\le\fc\|\vec\Ga-\vec\Ga'\|$
so that $\vec F$ is a strict contraction.
The claims are now a consequence of the contraction mapping theorem. 

\Item (b)
The first two bounds are special cases of part (a) with $B_j=0$.
Since $L_j$ is linear in its last $r$ arguments, $\de\vec\Ga=\vec\Ga-
\vec\Ga^{(1)}$ obeys
\begin{equation*}
\de \Ga_j(\vec\al)
=L_j\big(\vec\al\,,\,\de \vec\Ga(\vec\al)\big)
   +B_j\big(\vec\al\,,\,\vec\Ga^{(1)}(\vec\al)+\de\vec\Ga(\vec\al)\big)
\end{equation*}
for $1\le j\le r$. View this a fixed point equation determining 
$\vec\de\Ga$. The equation is of the form $\vec\de=\vec G(\vec\de)$
where
\begin{equation*}
\vec G(\vec\de)=\left[\begin{matrix}
              \tilde L_1(\vec\de) +\tilde B_1(\vec\Ga^{(1)}+\vec\de)\\
                           \vdots\\
                           \noalign{\vskip0.05in}
             \tilde L_r(\vec\de)+\tilde B_r(\vec\Ga^{(1)}+\vec\de)  
                      \end{matrix}\right]
\end{equation*}
If $\|\vec\de\|\le\fc$ then  $\|\vec\Ga^{(1)}+\vec\de\|\le 1$.
Therefore, by Corollary \ref{corSUBsubstitution}, 
$\vec G$ maps $\cB_\fc$ into $\cB_\fc$.
By Lemma \ref{lemSUBdiffnorm}, $\vec G$ is a strict contraction.
Apply the contraction mapping theorem. Since
$G_j(\vec 0)=\tilde B_j(\vec\Ga^{(1)})$ and 
\begin{equation*}
\|\vec\Ga^{(1)}\|\le\sfrac{1}{1-\fc}\|\vec f\,\|
\implies \tn \Ga^{(1)}_j\tn_w\le\sfrac{\|\vec f\,\|}{1-\fc}\la_j
\end{equation*}
for each $1\le j\le r$ and $B_j$ is of degree at least two in
its last $r$ arguments we have 
$\tn \tilde B_j(\vec\Ga^{(1)})\tn_w
     \le{\big(\sfrac{\|\vec f\|}{1-\fc}\big)}^2\tn B_j\tn_{w_{\ka,\la}}$
so that
$\|\vec G(\vec 0)\|\le{\big(\sfrac{\|\vec f\|}{1-\fc}\big)}^2
\max\limits_{1\le j\le r}\sfrac{1}{\la_j}\tn B_j\tn_{w_{\ka,\la}}$.
Therefore the fixed point $\vec\de=\de\vec\Ga$ obeys
\begin{equation*}
\|\de\vec\Ga\| \le \sfrac{1}{1-\fc} \|\vec G(\vec 0)\|
      \le \sfrac{\|\vec f\|^2}{(1-\fc)^3} \max_{1\le j\le r}\sfrac{1}{\la_j}\tn B_j\tn_{w_{\ka,\la}}
       \le(1-\fc) \max\limits_{1\le j\le r}\sfrac{1}{\la_j}\tn B_j\tn_{w_{\ka,\la}}
\end{equation*}
\end{proof}

\begin{example}\label{exSUBbackgroundfield}
We  assume that $X$ is a finite lattice of the form
$X = \fL_1/ \fL_2$, where $\fL_1$ is a lattice in $\bbbr^d$ and 
$\fL_2$ is a sublattice of $\fL_1$ of finite index. The Euclidean 
distance on $\bbbr^d$ induces a distance $|\,\cdot\,|$ on $X$.

\noindent
Let $W_1,W_2:X^3\rightarrow\bbbc$ and set, for complex fields $\phi_1,\phi_2$ on $X$
\begin{align*}
\cW_1(\phi_1,\phi_2)(x)= \sum_{y,z\in X}W_1(x,y,z)\,\phi_1(y)\,\phi_2(z) \\
\cW_2(\phi_1,\phi_2)(x)= \sum_{y,z\in X}W_2(x,y,z)\,\phi_1(y)\,\phi_2(z)
\end{align*}
Aso let $S_1$ and $S_2$ be two invertible operators on $L^2(X)$. 
Pretend that $S_1^{-1}$ and $S_2^{-1}$ are ``differential operators''.
Suppose that we are interested in solving
\begin{equation}\label{eqnSUBbgeqns}
\begin{split}
S_1^{-1}\phi_1 +\cW_1(\phi_1,\phi_2) &= \al_1\\
S_2^{-1}\phi_2+\cW_2(\phi_1,\phi_2)  &= \al_2
\end{split}
\end{equation}
for $\phi_1,\phi_2$ as functions of complex fields $\al_1,\al_2$. Suppose further that
we are thinking of the $\cW_j$'s as small. We would like to write the 
solution as a perturbation of the $\cW_1=\cW_2=0$ solution
$\phi_1=S_1\al_1$, $\phi_2=S_2\al_2$. So we substitute
\begin{equation*}
\phi_1=S_1\big(\al_1+\ga_1\big) \qquad 
\phi_2=S_2\big(\al_2+\ga_2\big) 
\end{equation*}
into \eqref{eqnSUBbgeqns}, giving
\begin{align*}
\ga_1 + \cW_1\big(S_1(\al_1+\ga_1)\,,\,S_2(\al_2+\ga_2)\big)&=0\\
\ga_2 + \cW_2\big(S_1(\al_1+\ga_1)\,,\,S_2(\al_2+\ga_2)\big)&=0
\end{align*}
This is of the  form \eqref{eqnSUBfixedpteqn} with
\begin{align*}
\vec f(\vec\al)
&=\left[\begin{matrix}
   - \cW_1\big(S_1\al_1\,,\,S_2\al_2\big)\\
   - \cW_2\big(S_1\al_1\,,\,S_2\al_2\big)
         \end{matrix}\right]
\\
\noalign{\vskip0.1in}
\vec L(\vec\al,\vec\ga)
&=\left[\begin{matrix}
         - \cW_1\big(S_1\ga_1\,,\,S_2\al_2\big)
         - \cW_1\big(S_1\al_1\,,\,S_2\ga_2\big)
        \\
         - \cW_2\big(S_1\ga_1\,,\,S_2\al_2\big)
         - \cW_2\big(S_1\al_1\,,\,S_2\ga_2\big)
       \end{matrix}\right]
\\
\noalign{\vskip0.1in}
\vec B(\vec\al,\vec\ga)(u)
&=\left[\begin{matrix}
         - \cW_1\big(S_1\ga_1\,,\,S_2\ga_2\big)\\
         - \cW_2\big(S_1\ga_1\,,\,S_2\ga_2\big)
         \end{matrix}\right]
\end{align*}
To  apply Proposition \ref{propSUBeqnsoln} to Example \ref{exSUBbackgroundfield}, 
fix any $\fm,\wf>0$ 
and use the norm $\tn \phi_j\tn$ with metric $\fm |\,\cdot\,|$ and 
weight factors $\wf$ to measure analytic maps like $\phi_j(\al_1,\al_2)$. 
See  Definition \ref{defSUBkrnel}.c. The weight factor $\wf$ is used for 
both $\al_1$ and $\al_2$. 
Like in \cite[\S IV]{CPR} and
\cite[Definition \defCPCoperatornorm]{CPC} we define, for any linear operator $S:L^2(X)\rightarrow L^2(X)$, 
the ``weighted'' $\ell^1$--$\ell^\infty$ norm
\begin{equation*}
\|S\|_{\fm}=\max\Big\{
      \sup_{y\in X}\sum_{x\in X}|S(x,y)|e^{\fm|y-x|}\ ,\ 
      \sup_{x\in X}\sum_{y\in X}|S(x,y)|e^{\fm|y-x|} \Big\}
\end{equation*}
\end{example}
 Proposition \ref{propSUBeqnsoln} can be applied to this situation:
\begin{corollary}\label{corSUBexample}
Let $K>0$. Write $\bar S= \max\limits_{j=1,2}\|S_j\|_{\fm}$
and $\bar W= \max\limits_{j=1,2}\|W_j\|_{\fm}$ and assume that 
\begin{equation*}
\bar S^2\bar W\,\wf < \min\big\{\sfrac{1}{12}\,,\,\sfrac{1}{2K}\big\}
\end{equation*}
Then there are field maps $\phi_1^{(\ge 2)}$, $\phi_2^{(\ge 2)}$ such that
\begin{align*}
\phi_1(\al_1,\al_2)
     &=S_1\al_1
      +\phi_1^{(\ge 2)}(\al_1,\al_2)\\
\phi_2(\al_1,\al_2)
     &=S_2\al_2
      +\phi_2^{(\ge 2)}(\al_1,\al_2)
\end{align*}
solves the equations \eqref{eqnSUBbgeqns} of  Example \ref{exSUBbackgroundfield} 
and obeys
\begin{equation*}
\TN \phi_{j}^{(\ge 2)}\TN
\le 2\bar S^3\,\bar W\,\wf^2
\end{equation*}
Furthermore 
$\phi_j^{(\ge 2)}$ is of degree at least two in $(\al_1,\al_2)$. 
The solution is unique in 
\begin{equation*}
\set{(\phi_1,\phi_2)\in L^2(X)\times L^2(X)}
           {\tn S_1^{-1}\phi_1\tn,
             \tn S_2^{-1}\phi_2\tn\le K\wf}
\end{equation*}
\end{corollary}

\begin{proof}
In Example \ref{exSUBbackgroundfield} we wrote the equations 
\eqref{eqnSUBbgeqns} in the form
\begin{equation}\label{eqnSUBexeqn}
\vec\ga
=\vec f(\vec\al)+\vec L(\vec\al,\vec\ga)+\vec B\big(\vec\al,\vec\ga\big)
\end{equation}
Now apply Proposition \ref{propSUBeqnsoln}.a and Remark \ref{remSUBcB}.a 
with $r=s=2$ and
\begin{equation*}
d_{\rm max}=2\qquad
\fc=\half\qquad
\ka_1= \ka_2 = \la_1=\la_2=\wf
\end{equation*}
Since
\begin{align*}
\tn f_j\tn_w &\le  \|S_1\|_\fm\|S_2\|_\fm\|W_j\|_{\fm}\ \ka_1\ka_2\cr
\tn L_j\tn_{w_{\ka,\la}}
  &\le \|S_1\|_\fm\|S_2\|_\fm\|W_j\|_{\fm}\ \big(\la_1\ka_2
            +\ka_1\la_2\big) \\
\tn B_j\tn_{w_{\ka,\la}}
    &\le \|S_1\|_\fm\|S_2\|_\fm\|W_j\|_{\fm}\la_1\la_2
\end{align*}
By hypothesis, $\tn f_j\tn_w,\,\tn L_j\tn_{w_{\ka,\la}},
\,\tn B_j\tn_{w_{\ka,\la}}<\sfrac{1}{6}\la_j$
and Proposition \ref{propSUBeqnsoln}.a gives a solution $\vec\Ga(\vec\al)$
to \eqref{eqnSUBexeqn} that obeys the bound
\begin{equation*}
\tn \Ga_j\tn_w\le 2\|S_1\|_\fm\|S_2\|_\fm\|W_j\|_{\fm}\wf^2
\end{equation*}
Setting
\begin{alignat*}{3}
\phi_1(\al_1,\al) 
&=S_1\al_1 +S_1  \Ga_1(\al_1,\al_2) &\qquad
\phi^{(\ge 2)}_1(\al_1,\al_2) &= S_1  \Ga_1(\al_1,\al_2) \\
\phi_2(\al_1,\al) 
&=S_2\al_2 +S_2\Ga_2(\al_1,\al_2)&\qquad
\phi^{(\ge 2)}_2(\al_1,\al_2) &= S_2 \Ga_2(\al_1,\al_2) 
\end{alignat*}
we have all of the claims, except for uniqueness.

We now prove uniqueness. Assume that 
$\phi_j=S_j\Phi_j$ and that
$\phi_j=S_j(\Phi_j+\de\Phi_j)$ 
both solve \eqref{eqnSUBbgeqns}, 
with $\tn \Phi_j+\de\Phi_j\tn \le K\wf$ and
with $S_j\Phi_j$ being the solution constructed above.
Then $\de\Phi_j$ is a solution of
\begin{align*}
\de\Phi_1 &=-\cW_1\big(S_1(\Phi_1+\de\Phi_1)\,,\,
                      S_2(\Phi_2+\de\Phi_2)\big)
          +\cW_1\big(S_1\Phi_*\,,\,
                      S_2\Phi\big)\\
\de\Phi_2&=-\cW_2\big(S_2(\Phi_2+\de\Phi_2)\,,\,
                      S_1(\Phi_1+\de\Phi_1)\big)
       +\cW_2\big(S_2\Phi\,,\,
                      S_1\Phi_1\big)\cr
\end{align*}
Since
\begin{equation*}
\TN \cW_j\big(S_1\al_1\,,\, S_2\al_2\big) \TN
\le \|W_j\|_\fm\,\tn S_1\al_1\tn\,\tn S_2\al_2\tn
\le \|W_j\|_\fm\,\|S_1\|_\fm\|S_2\|_\fm\,\tn \al_1\tn\,\tn \al_2\tn
\end{equation*}
we have
\begin{align*}
\tn \de\Phi_1\tn &\le  \|W_1\|_\fm\,\|S_1\|_\fm\,\|S_2\|_\fm\,
                       \big\{\tn \de\Phi_1\tn\,\tn \Phi_2+\de\Phi_2\tn
                          +\tn \Phi_1\tn\,\tn \de\Phi_2\tn\big\} \\
\tn \de\Phi_2\tn &\le  \|W_2\|_\fm\,\|S_1\|_\fm\,\|S_2\|_\fm\,
                       \big\{\tn \de\Phi_1\tn\,\tn \Phi_2+\de\Phi_2\tn
                          +\tn \Phi_1\tn\,\tn \de\Phi_2\tn\big\} \cr
\end{align*}
By hypothesis
\begin{equation*}
\tn \Phi_1\tn\le \wf+2\|S_1\|_\fm\|S_2\|_\fm\|W_j\|_{\fm}\wf^2
          \le\sfrac{7}{6}\wf\qquad
\tn \Phi_2+\de\Phi_2\tn\le K\wf
\end{equation*}
so that
\begin{align*}
\tn \de\Phi_1\tn+\tn \de\Phi_2\tn
  &\le \big(\|W_1\|_\fm+\|W_2\|_\fm\big)\,\|S_1\|_\fm\,\|S_2\|_\fm\,
      \max\big\{\sfrac{7}{6},K\big\}\wf\ 
      \big(\tn \de\Phi_1\tn+\tn \de\Phi_2\tn\big) \\
  &\le \bar S^2\bar W \wf\ 
      2\max\big\{\sfrac{7}{6},K\big\}\ 
      \big(\tn \de\Phi_1\tn+\tn \de\Phi_2\tn\big)
\end{align*}
thereby forcing $\tn \de\Phi_*\tn = \tn \de\Phi\tn=0$.
\end{proof}

\newpage
\appendix
\section{A Generalisation of Young's Inequality}
\begin{lemma}\label{lemSUBgenLoneLinfty}
Let $n\in\bbbn$. For each $1\le\ell\le n$, let
\begin{itemize}[leftmargin=*, topsep=2pt, itemsep=0pt, parsep=0pt]
\item $(X_\ell,d\mu_\ell)$ be a measure space, 
\item $f_\ell:X_\ell\rightarrow\bbbc$ be measureable and
\item $p_\ell\in (0,\infty]$.
\end{itemize}

\noindent
Let $K:\cartprod\limits_{\ell=1}^n X_\ell\rightarrow\bbbc$ have 
finite $L^1$--$L^\infty$ norm and assume that
$\sum\limits_{\ell=1}^n\sfrac{1}{p_\ell}=1$. Then
\begin{align*}
\bigg|\int_{\cartprod\limits_{\ell=1}^n X_\ell} 
             K(x_1,\cdots,x_n)\smprod_{\ell=1}^n f_\ell(x_\ell)
             \smprod_{\ell=1}^n d\mu_\ell(x_\ell)\bigg|
&\le \|K\|_{L^1-L^\infty} \smprod_{\ell=1}^n \|f_\ell\|_{L^{p_\ell}(d\mu_\ell)}
\end{align*}
\end{lemma}
\begin{proof}
 We'll use the short hand notations 
$dm(x_1,\cdots,x_n)= \prod\limits_{\ell=1}^n d\mu_\ell(x_\ell)$
and $X= \cartprod\limits_{\ell=1}^n X_\ell $. By H\"older (with the usual
interpretations when some $p_\ell=\infty$),
\begin{align*}
&\bigg|\int_X K(x_1,\cdots,x_n)\smprod_{\ell=1}^n f_\ell(x_\ell)\ 
             dm(x_1,\cdots,x_n)\bigg| \\
&\hskip0.5in\le 
  \int_X \smprod_{\ell=1}^n\big\{\big|K(x_1,\cdots,x_n)\big|^{1/p_\ell}\
                                          |f_\ell(x_\ell)|\big\}\ 
             dm(x_1,\cdots,x_n)\\
&\hskip0.5in=  \smprod_{\ell=1}^n
  \bigg[\int_X \big|K(x_1,\cdots,x_n)\big|\,|f_\ell(x_\ell)|^{p_\ell}\ 
              \smprod_{\ell=1}^n d\mu_\ell(x_\ell)\bigg]^{1/p_\ell} \cr
&\hskip0.5in\le  \smprod_{\ell=1}^n
  \bigg[\|K\|_{L^1-L^\infty} \int_{X_\ell}|f_\ell(x_\ell)|^{p_\ell}\ 
               d\mu_\ell(x_\ell)\bigg]^{1/p_\ell} \\
&\hskip0.5in =\|K\|_{L^1-L^\infty} 
           \smprod_{\ell=1}^n \|f_\ell\|_{L^{p_\ell}(d\mu_\ell)}
\end{align*}
\end{proof}

\newpage
\bibliographystyle{plain}
\bibliography{refs}

\begin{thebibliography}{10}

\bibitem{BalLausane}
T.~Balaban.
\newblock {The Ultraviolet Stability Bounds for Some Lattice $\si$--Models and
  Lattice Higgs--Kibble Models}.
\newblock In {\em {Proc. of the International Conference on Mathematical
  Physics, Lausanne, 1979}}, pages 237--240. Springer, 1980.

\bibitem{BalPalaiseau}
T.~Balaban.
\newblock A low temperature expansion and {``}spin wave picture{''} for
  classical {$N$}-vector models.
\newblock In {\em {Constructive Physics (Palaiseau, 1994), Lecture Notes in
  Physics, 446}}, pages 201--218. Springer, 1995.

\bibitem{CPR}
T.~Balaban, J.~Feldman, H.~Kn{\"o}rrer, and E.~Trubowitz.
\newblock {Power Series Representations for Bosonic Effective Actions}.
\newblock {\em Journal of Statistical Physics}, 134:839--857, 2009.

\bibitem{CPC}
T.~Balaban, J.~Feldman, H.~Kn{\"o}rrer, and E.~Trubowitz.
\newblock {Power Series Representations for Complex Bosonic Effective Actions.
  I. A Small Field Renormalization Group Step}.
\newblock {\em Journal of Mathematical Physics}, 51:053305, 2010.

\bibitem{CPS}
T.~Balaban, J.~Feldman, H.~Kn{\"o}rrer, and E.~Trubowitz.
\newblock {Power Series Representations for Complex Bosonic Effective Actions.
  II. A Small Field Renormalization Group Flow}.
\newblock {\em Journal of Mathematical Physics}, 51:053306, 2010.

\bibitem{UV}
T.~Balaban, J.~Feldman, H.~Kn{\"o}rrer, and E.~Trubowitz.
\newblock {The Temporal Ultraviolet Limit for Complex Bosonic Many-body
  Models}.
\newblock {\em Annales Henri Poincar{\'e}}, 11:151--350, 2010.

\bibitem{ParOv}
T.~Balaban, J.~Feldman, H.~Kn{\"o}rrer, and E.~Trubowitz.
\newblock {Complex Bosonic Many--body Models: Overview of the Small Field
  Parabolic Flow}.
\newblock Preprint, 2016.

\bibitem{PAR1}
T.~Balaban, J.~Feldman, H.~Kn{\"o}rrer, and E.~Trubowitz.
\newblock {The Small Field Parabolic Flow for Bosonic Many--body Models: Part 1
  --- Main Results and Algebra}.
\newblock Preprint, 2016.

\bibitem{PAR2}
T.~Balaban, J.~Feldman, H.~Kn{\"o}rrer, and E.~Trubowitz.
\newblock {The Small Field Parabolic Flow for Bosonic Many--body Models: Part 2
  --- Fluctuation Integral and Renormalization}.
\newblock Preprint, 2016.

\bibitem{Dim1}
J.~Dimock.
\newblock {The renormalization group according to Balaban -- I. small fields}.
\newblock {\em Reviews in Mathematical Physics}, 25:1--64, 2013.

\bibitem{GK}
K.~Gawedzki and A.~Kupiainen.
\newblock {A rigorous block spin approach to massless lattice theories}.
\newblock {\em Comm. Math. Phys.}, 77:31--64, 1980.

\bibitem{KAD}
L.P. Kadanoff.
\newblock {Scaling laws for Ising models near $T_c$}.
\newblock {\em Physics}, 2:263, 1966.

\bibitem{LL}
E.~H. Lieb and M.~Loos.
\newblock {\em Analysis}.
\newblock American Mathematical Society, 1996.

\end{thebibliography}

\end{document}